\documentclass[runningheads,a4paper]{llncs}
\usepackage{graphicx}
\usepackage{psfrag}
\usepackage{amsmath}
\usepackage[varg]{txfonts}
\usepackage[subrefformat=parens]{subcaption}
\usepackage{hyperref}

\title{Recognizing Simple-Triangle Graphs by Restricted 2-Chain Subgraph Cover}
\author{Asahi Takaoka}
\institute{
  Department of Information Systems Creation, 
  Kanagawa University, \\ 
  Rokkakubashi 3-27-1 Kanagawa-ku, 
  Kanagawa, 221--8686, Japan \\
  takaoka@jindai.jp
}

\begin{document}
\maketitle

\begin{abstract}
A simple-triangle graph (also known as a PI graph) is 
the intersection graph of a family of triangles 
defined by a point on a horizontal line and 
an interval on another horizontal line. 
The recognition problem for simple-triangle graphs 
was a longstanding open problem, and 
recently a polynomial-time algorithm has been given 
[G. B. Mertzios, The Recognition of Simple-Triangle Graphs and 
of Linear-Interval Orders is Polynomial, 
SIAM J. Discrete Math., 29(3):1150--1185, 2015]. 
Along with the approach of this paper, 
we show a simpler recognition algorithm 
for simple-triangle graphs. 
To do this, we provide a polynomial-time algorithm 
to solve the following problem: 
Given a bipartite graph $G$ and a set $F$ of edges of $G$, 
find a 2-chain subgraph cover of $G$ 
such that one of two chain subgraphs has no edges in $F$. 
\end{abstract}

\section{Introduction}
Let $L_1$ and $L_2$ be two horizontal lines in the plane 
with $L_1$ above $L_2$. 
A \emph{simple-triangle graph} 
is the intersection graph of a family of triangles 
spanned by a point on $L_1$ and an interval on $L_2$. 
That is, a simple undirected graph is called a simple-triangle graph 
if there is such a triangle for each vertex 
and two vertices are adjacent if and only if 
the corresponding triangles have a nonempty intersection. 
See Figure~\ref{fig:example}\subref{fig:graph} 
and~\ref{fig:example}\subref{fig:representation} for example. 
Simple-triangle graphs are also known as 
\emph{PI graphs}~\cite{COS08-ENDM,CK87-CN}, 
where \emph{PI} stands for \emph{Point-Interval}. 
Simple-triangle graphs were introduced in~\cite{CK87-CN} 
as a generalization of 
both interval graphs and permutation graphs. 
Simple-triangle graphs are also known as 
a proper subclass of trapezoid graphs~\cite{CK87-CN,DGP88-DAM}, 
another generalization of interval graphs and permutation graphs. 

Recently, the graph isomorphism problem for trapezoid graphs
has shown to be isomorphism-complete~\cite{Takaoka15-IEICE} 
(that is, polynomial-time equivalent to the problem for general graphs). 
Since the problem can be solved in linear time 
for interval graphs~\cite{LB79-JACM} and 
for permutation graphs~\cite{Colbourn81-Networks}, 
it has become an interesting question to give 
the structural characterization of graph classes lying strictly 
between permutation graphs and trapezoid graphs or 
between interval graphs and trapezoid graphs~\cite{Uehara14-DMTCS}. 
Although a lot of research has been done for interval graphs, 
for permutation graphs, and for trapezoid graphs 
(see~\cite{Spinrad03} for example), 
there are few results for simple-triangle graphs~\cite{BLR10-Order,COS08-ENDM,CK87-CN}. 
It is only recently that a polynomial-time recognition algorithm 
have been given~\cite{Mertzios13-LNCS,Mertzios15-SIAMDM}. 

The recognition algorithm first reduces the recognition problem 
to the \emph{linear-interval cover} problem. 
The algorithm then reduces the linear-interval cover problem 
to \emph{gradually mixed} formulas, 
a tractable subclass of 3-satisfiability (3SAT). 
Finally, the algorithm solves the gradually mixed formulas 
by reducing it to 2-satisfiability (2SAT), 
which can be solved in linear time (see~\cite{APT79-IPL} for example). 
The total running time of the algorithm is $O(n^2\bar{m})$, 
where $n$ and $\bar{m}$ are the number of vertices and non-edges 
of the given graph, respectively. 

In this paper, 
we introduce the \emph{restricted 2-chain subgraph cover} problem 
as a generalization of the linear-interval cover problem. 
Then, we show that our problem is directly reducible to 2SAT. 
This result does not improve the running time, 
but it can simplify the previous algorithm 
for the recognition of simple-triangle graphs. 

\begin{figure*}[t]
  \psfrag{L1}{$L_1$}
  \psfrag{L2}{$L_2$}
  \psfrag{a1}{$a_1$}
  \psfrag{a2}{$a_2$}
  \psfrag{b1}{$b_1$}
  \psfrag{b2}{$b_2$}
  \psfrag{c1}{$c_1$}
  \psfrag{c2}{$c_2$}
  \psfrag{c3}{$c_3$}
  \psfrag{a1'}{$a_1'$}
  \psfrag{a2'}{$a_2'$}
  \psfrag{b1'}{$b_1'$}
  \psfrag{b2'}{$b_2'$}
  \psfrag{c1'}{$c_1'$}
  \psfrag{c2'}{$c_2'$}
  \psfrag{c3'}{$c_3'$}
  \centering\subcaptionbox{A graph $G$. \label{fig:graph}}{\includegraphics[scale=0.6]{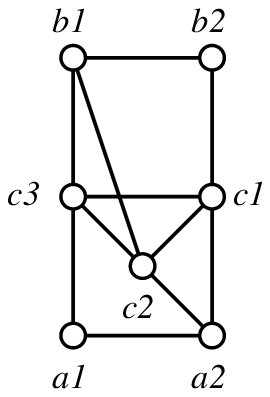}}
  \centering\subcaptionbox{A representation of $G$. \label{fig:representation}}{\includegraphics[scale=0.6]{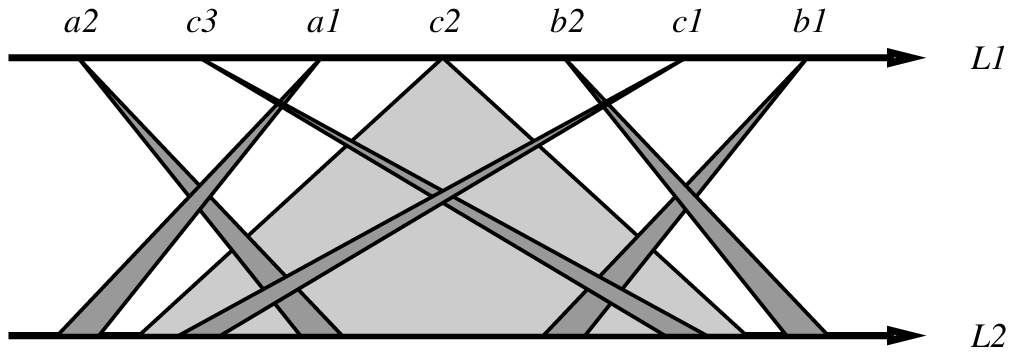}}
  \centering\subcaptionbox{The order $P$. \label{fig:order}}{\includegraphics[scale=0.6]{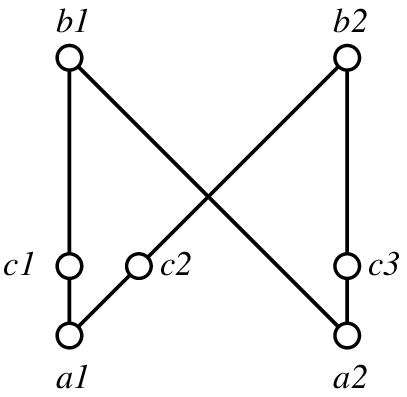}}
  \caption{
    A simple-triangle graph $G$, 
    an intersection representation of $G$, and 
    the Hasse diagram of the linear-interval order $P$ obtained from $G$. 
    }
  \label{fig:example}
\end{figure*}

\subsection{Linear-Interval Cover}
In this section, we briefly describe the linear-interval cover problem and 
the reduction to it from the recognition problem for simple-triangle graphs. 
See~\cite{Mertzios15-SIAMDM} for the details. 
We first show that the recognition of simple-triangle graphs is reducible 
to that of linear-interval orders in $O(n^2)$ time, 
where $n$ is the number of vertices of the given graph. 
A \emph{partial order} is a pair $P = (V, \prec)$, 
where $V$ is a finite set and $\prec$ is a binary relation on $V$ 
that is irreflexive and transitive. 
Partial orders are represented by \emph{transitively oriented graphs}, 
which are directed graphs such that 
if $u \rightarrow v$ and $v \rightarrow w$, then $u \rightarrow w$ 
for any three vertices $u, v, w$ of the graphs. 

There is a correspondence between partial orders and 
the intersection graphs of geometric objects spanned 
between two horizontal lines $L_1$ and $L_2$~\cite{GRU83-DM}. 
A partial order $P = (V, \prec)$ is called 
a \emph{linear-interval order}~\cite{BLR10-Order,COS08-ENDM} 
if for each element $v \in V$, there is a triangle $T_v$ 
spanned by a point on $L_1$ and an interval on $L_2$ 
such that $u \prec v$ if and only if $T_u$ lies completely to the left of $T_v$ 
for any two elements $u, v \in V$. 
See Figure~\ref{fig:example}\subref{fig:representation} 
and~\ref{fig:example}\subref{fig:order} for example. 

For a graph $G = (V, E)$, the graph $\overline{G} = (V, \overline{E})$ 
is called the \emph{complement} of $G$, 
where $uv \in \overline{E}$ if and only if  $uv \notin E$ 
for any pair of vertices $u, v \in V$. 
We can obtain a linear-interval order from a simple-triangle graph $G$ 
by giving a transitive orientation to the complement $\overline{G}$ of $G$. 
The complement $\overline{G}$ might have some different transitive orientations, 
but the following theorem states that 
any transitive orientation of $\overline{G}$ gives a linear-interval order 
if $G$ is a simple-triangle graph. 
A property of partial orders is said to be 
a \emph{comparability invariant} if either all orders obtained 
from the same graph have that property or none have that property. 
\begin{theorem}[\cite{COS08-ENDM}]
Being a linear-interval order is a comparability invariant. 
\end{theorem}
Many algorithms have been proposed for transitive orientation, 
including a linear-time one~\cite{MS99-DM}. 
Since the complement of a graph can be obtained in $O(n^2)$ time, 
the recognition of simple-triangle graphs is reducible 
to that of linear-interval orders in $O(n^2)$ time. 

We then show that the recognition of linear-interval orders is reducible 
to the linear-interval cover problem in $O(n^2)$ time, 
where $n$ is the number of elements of the given partial orders. 
Let $P = (V, \prec)$ be a partial order 
with $V = \{ v_1, v_2, \ldots, v_n \}$, 
and let $V' = \{ v_1', v_2', \ldots, v_n' \}$. 
The \emph{domination bipartite graph} $C(P) = (V, V', E)$ of $P$ is defined 
such that $v_iv_j' \in E$  if and only if $v_i \prec v_j$ in $P$~\cite{MS94-JAL}. 
We also define that $E_0 = \{ v_iv_i' \mid v_i \in V \}$. 
The \emph{bipartite complement} of $C(P)$ is 
the bipartite graph $\widehat{C(P)} = (V, V', \hat{E})$, 
where $\hat{E}$ is the set of non-edges between the vertices of $V$ and $V'$, 
that is, $v_iv_j' \in \hat{E}$ if and only if  $v_iv_j' \notin E$ 
for any vertices $v_i \in V$ and $v_j' \in V'$. 
By definition, we have $E_0 \subseteq \hat{E}$. 

Let $2K_2$ denote a graph consisting of 
four vertices $u_1, u_2, v_1, v_2$ with two edges $u_1v_1, u_2v_2$. 
A bipartite graph $G = (U, V, E)$ 
is called a \emph{chain graph}~\cite{Yannakakis82-SIAM} 
if it has no $2K_2$ as an induced subgraph. 
Equivalently, a bipartite graph $G$ is a chain graph if and only if 
there is a linear ordering $u_1, u_2, \ldots, u_n$ on $U$ (or $V$) 
such that $N_G(u_1) \subseteq N_G(u_2) \subseteq \ldots \subseteq N_G(u_n)$, 
where $N_G(u)$ is the set of vertices adjacent to $u$ in $G$. 
A \emph{chain subgraph} of $G$ is a subgraph of $G$ that has no induced $2K_2$. 
A bipartite graph $G = (U, V, E)$ is said to be \emph{covered} by 
two chain subgraphs $G_1 = (U, V, E_1)$ and $G_2 = (U, V, E_2)$ 
if $E = E_1 \cup E_2$ 
(we note that in general, $E_1$ and $E_2$ are not disjoint), and 
the pair of chain subgraphs $(G_1, G_2)$ is called 
a \emph{2-chain subgraph cover} of $G$. 
For a partial order $P$, 
a 2-chain subgraph cover $(G_1, G_2)$ of $\widehat{C(P)}$ is called 
a \emph{linear-interval cover} if $G_1$ has no edges in $E_0$. 
\begin{theorem}[\cite{Mertzios15-SIAMDM}]
A partial order $P$ is linear-interval order if and only if 
$\widehat{C(P)}$ has a linear-interval cover. 
\end{theorem}
The \emph{linear-interval cover} problem asks 
whether $\widehat{C(P)}$ has a linear-interval cover. 
Since $C(P)$ and $\widehat{C(P)}$ can be obtained in $O(n^2)$ time 
from a partial order $P$, 
the recognition of linear-interval orders is reducible 
to the linear-interval cover problem in $O(n^2)$ time.

\subsection{Restricted 2-Chain Subgraph Cover}
As a generalization of the linear-interval cover problem, 
we consider the following restricted problem 
for 2-chain subgraph cover. 
\medskip
\begin{center}
\fbox{\parbox{0.95\linewidth}{\noindent
{\sc Restricted 2-Chain Subgraph Cover}\\[.8ex]
\begin{tabular*}{.95\textwidth}{rl}
\emph{Instance:} & A bipartite graph $G = (U, V, E)$ and 
a set $F$ of edges of $G$.\\
\emph{Question:} & Find a 2-chain subgraph cover $(G_1, G_2)$ of $G$ \\
& such that $G_1$ has no edges in $F$. 
\end{tabular*}
}}
\end{center}
\medskip
Notice that $G_2$ has all the edges in $F$. 
Let $\hat{E}$ be the set of edges of the bipartite complement $\hat{G}$ of $G$ . 
Let $m = |E|$, $\hat{m} = |\hat{E}|$, and $f = |F|$. 
The following is our main result. 
\begin{theorem}
The restricted 2-chain subgraph cover problem can be solved 
in $O(m\hat{m} + \min\{m^2, \hat{m}(\hat{m}+f)\})$ time. 
\end{theorem}

\begin{figure*}[t]
  \psfrag{u1}{$u_1$}
  \psfrag{u2}{$u_2$}
  \psfrag{v1}{$v_1$}
  \psfrag{v2}{$v_2$}
  \centering
  \begin{tabular}{ccccc}
    \includegraphics[scale=.42]{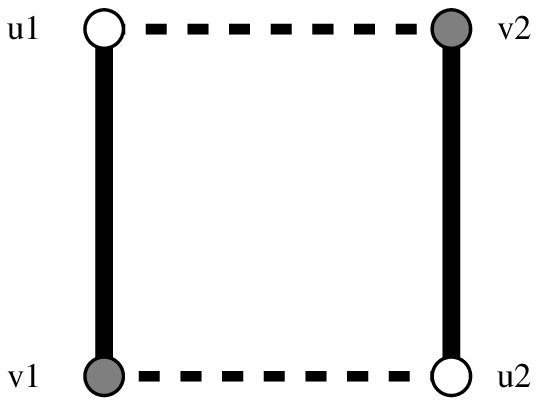}&
    \includegraphics[scale=.42]{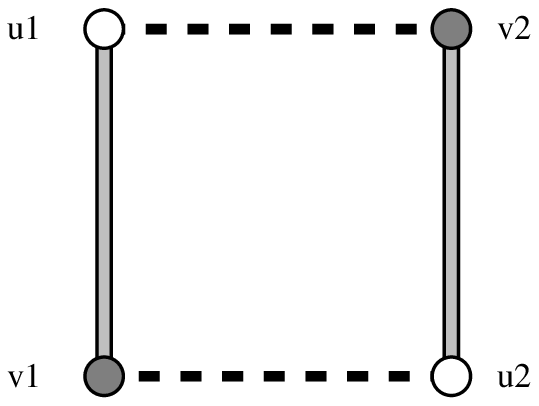}&
    \includegraphics[scale=.42]{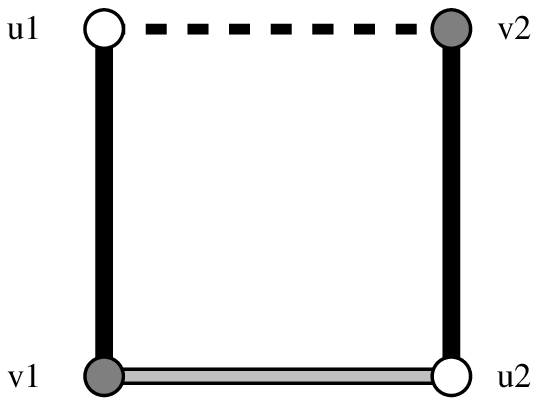}&
    \includegraphics[scale=.42]{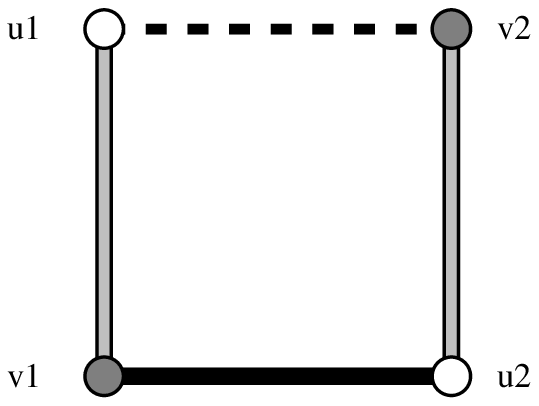}&
    \includegraphics[scale=.42]{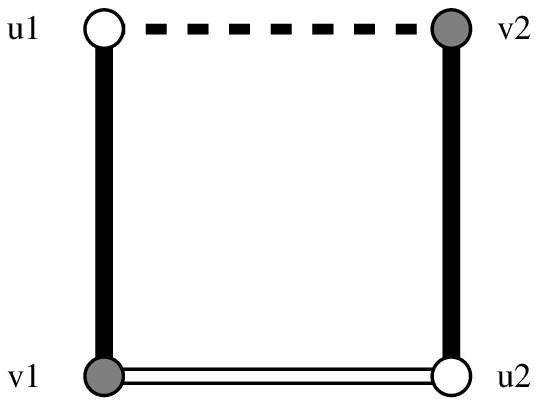}
    \\
    $(A_1)$ & $(A_2)$ & $(B_1)$ & $(B_2)$ & $(C)$ 
  \end{tabular}
  \caption{
    Forbidden configurations. 
    Solid lines and gray solid lines denote edges in $E_r$ and $E_b$, 
    respectively. 
    Dashed lines denote non-edges in $\hat{E}$, and 
    double lines denote edges in $F$. } 
  \label{fig:A-C}
\end{figure*}

In the rest of this section, we describe the outline of our algorithm. 
The details are shown in Section~\ref{sect:algorithm}. 
Two edges $e$ and $e'$ of a bipartite graph $G = (U, V, E)$ is said to 
be \emph{in conflict in $G$} 
if the vertices of $e$ and $e'$ induce a $2K_2$ in $G$. 
An edge $e \in E$ is said to be \emph{committed} 
if there is another edge $e' \in E$ such that 
$e$ and $e'$ are in conflict in $G$, 
and said to be \emph{uncommitted} otherwise. 
Let $E_c$ be the set of committed edges of $G$, and 
let $E_u$ be the set of uncommitted edges of $G$. 

Suppose $G$ has a 2-chain subgraph cover $(G_1, G_2)$ 
such that $G_1$ has no edges in $F$. 
If two edges $e, e' \in E$ are in conflict in $G$, 
then $e$ and $e'$ may not belong to the same chain subgraph. 
Therefore, each committed edge in $E_c$ belongs to either $G_1$ or $G_2$. 
We refer to the committed edges of $G_1$ as \emph{red} edges 
and the committed edges of $G_2$ as \emph{blue} edges. 
Let $E_r$ and $E_b$ be the set of red edges and blue edges, respectively, 
and we call $(E_r, E_b)$ the \emph{bipartition} of $E_c$. 
Notice that $F \subseteq E_b \cup E_u$ since $E_r$ has no edges in $F$. 
Hence, we assume without explicitly stating it in the rest of this paper that 
all the committed edges in $F$ are in $E_b$. 
We can also see that 
the bipartition $(E_r, E_b)$ does not have the following 
\emph{forbidden configurations} (see Figure~\ref{fig:A-C}). 
\begin{itemize}
\item 
Configuration $(A_1)$ [resp., $(A_2)$] consists of four vertices $u_1, u_2 \in U$ and 
$v_1, v_2 \in V$ with edges $u_1v_1, u_2v_2 \in E_r$ [resp., $u_1v_1, u_2v_2 \in E_b$] 
and non-edges $u_1v_2, u_2v_1 \in \hat{E}$, that is, 
$u_1v_1$ and $u_2v_2$ are in conflict in $G$; 
\item 
Configuration $(B_1)$ [resp., $(B_2)$] consists of four vertices $u_1, u_2 \in U$ and 
$v_1, v_2 \in V$ with edges $u_1v_1, u_2v_2 \in E_r$ [resp., $u_1v_1, u_2v_2 \in E_b$], 
a non-edge $u_1v_2 \in \hat{E}$, and 
an edge $u_2v_1 \in E_b$ [resp., $u_2v_1 \in E_r$]; 
\item 
Configuration $(C)$ consists of four vertices $u_1, u_2 \in U$ and 
$v_1, v_2 \in V$ with edges $u_1v_1, u_2v_2 \in E_r$, 
a non-edge $u_1v_2 \in \hat{E}$, and an edge $u_2v_1 \in F$. 
\end{itemize}

Our algorithm construct a bipartition $(E_r, E_b)$ of $E_c$ 
that does not have some forbidden configurations. 
A bipartition of $E_c$ is called \emph{$(A, C)$-free} 
if it has neither configuration $(A_1)$, $(A_2)$, nor $(C)$. 
A bipartition of $E_c$ is called \emph{$(A, B, C)$-free} 
if it has neither configuration $(A_1)$, $(A_2)$, $(B_1)$, $(B_2)$, nor $(C)$. 

\begin{theorem}
A bipartite graph $G$ has a 2-chain subgraph cover $(G_1, G_2)$ such that 
$G_1$ has no edges in $F$ if and only if 
$E_c$ has an $(A, C)$-free bipartition. 
\end{theorem}

The outline of our algorithm is as follows. 
\begin{description}
\item[Step 1:] 
Partition the set $E_c$ of committed edges into 
an $(A, C)$-free bipartition $(E_r, E_b)$ by solving 2SAT. 
\item[Step 2:] 
From the $(A, C)$-free bipartition $(E_r, E_b)$ of $E_c$, 
compute an $(A, B, C)$-free bipartition $(E_r', E_b')$ of $E_c$
by swapping some edges between $E_r$ and $E_b$. 
\item[Step 3:] 
From the $(A, B, C)$-free bipartition $(E_r', E_b')$ of $E_c$, 
compute a desired 2-chain subgraph cover of $G$ 
by adding some uncommitted edges into $E_r'$ and $E_b'$. 
\end{description}
We will show in Sections~\ref{sect:bipartition},~\ref{sect:swap}, and~\ref{sect:add} 
that \textbf{Step 1}, \textbf{Step 2}, and \textbf{Step 3} can be done 
in $O(\min\{m^2, \hat{m}(\hat{m}+f)\})$ time, $O(m\hat{m})$ time, and 
linear time, respectively.

\subsection{Related Work}
A bipartite graph $G = (U, V, E)$ is said to be \emph{covered} by $k$ subgraphs 
$G_i = (U, V, E_i)$, $1 \leq i \leq k$, 
if $E = E_1 \cup E_2 \cup \cdots \cup E_k$. 
A \emph{$k$-chain subgraph cover} problem asks whether 
a given bipartite graph can be covered by $k$ chain subgraphs. 
The $k$-chain subgraph cover problem is NP-complete if $k \geq 3$, 
while it is polynomial-time solvable 
if $k \leq 2$~\cite{Yannakakis82-SIAM}. 

The 2-chain subgraph cover problem is closely related to 
some recognition problems; more precisely, 
they can be efficiently reduced to the 2-chain subgraph cover problem. 
They are the recognition problems for 
threshold dimension~2 graphs on split graphs~\cite{IP81-ADM,RS95-STOC}, 
circular-arc graphs with clique cover number~2~\cite{HH97-GaC,Spinrad88-JCTSB}, 
2-directional orthogonal ray graphs~\cite{STU10-DAM,TTU14-IEICE}, and 
trapezoid graphs~\cite{MS94-JAL}. 
Other related problems and surveys can be found in 
Chapter~8 of~\cite{MP95-book} and Section 13.5 of~\cite{Spinrad03}. 

As far as we know, there are two approaches 
for the 2-chain subgraph cover problem and 
the other related problems. 
One approach is shown in~\cite{MS94-JAL,Spinrad88-JCTSB}, 
which reduces the 2-chain subgraph cover problem 
to the recognition of 2-dimensional partial orders. 
This approach is used in the fastest known algorithm~\cite{MS94-JAL} 
with a running time of $O(n^2)$, 
where $n$ is the number of vertices of the given graph. 
Another approach can be found in~\cite{HH97-GaC,IP81-ADM,RS95-STOC}. 
They show that 
a bipartite graph $G = (U, V, E)$ has a 2-chain subgraph cover if and only if 
the \emph{conflict graph} $G^* = (V^*, E^*)$ of $G$ is bipartite, 
where $V^* = E$ and two edges $e$ and $e'$ in $E$ are adjacent in $G^*$ 
if $e$ and $e'$ are in conflict in $G$. 
We note that the algorithm in this paper is based on the latter approach. 

In Section 8.6 of~\cite{MP95-book}, the following problem is considered 
for recognizing threshold dimension~2 graphs: 
Given a bipartite graph $G$ and a pair $(F_1, F_2)$ of edge sets, 
find a 2-chain subgraph cover $(G_1, G_2)$ of $G$ 
such that $G_1$ and $G_2$ have every edge in $F_1$ and $F_2$, respectively. 
We call such a problem the \emph{extension problem} for 2-chain subgraph cover. 
We emphasize that the extension problem is not a generalization of 
our restricted 2-chain subgraph cover problem 
since in the extension problem, 
$G_1$ and $G_2$ are allowed to have all the uncommitted edges of $G$. 
As shown in~\cite{MP95-book}, 
this problem can be solved in polynomial time by reducing it 
to some variation of the recognition problem 
for 2-dimensional partial orders. 
We note that this variation can be stated as 
the problem of extending a partial orientation of a permutation graph to 
a 2-dimensional partial order~\cite{KKKW12-LNCS}.

\section{Algorithm}\label{sect:algorithm}

\subsection{Partitioning Edges}\label{sect:bipartition}
A \emph{2CNF formula} is a Boolean formula in conjunctive normal form 
with at most two literals per clause. In this section, 
we construct a 2CNF formula $\phi$ 
such that $\phi$ is satisfiable if and only if 
$G$ has an $(A, C)$-free bipartition of $E_c$. 
The construction of $\phi$ is as follows: 
\begin{itemize}
\item 
Assign the Boolean variable $x_e$ to each committed edge $e \in E_c$; 
\item 
Add the clause $(x_e)$ for each edge $e \in F \cap E_c$; 
\item 
For each pair of two edges $e$ and $e'$ in $E_c$, 
add the clauses $(x_e \vee x_{e'})$ and 
$(\overline{x_e} \vee \overline{x_{e'}})$ to $\phi$ 
if $e$ and $e'$ are in conflict in $G$; 
\item 
For each pair of two edges $e$ and $e'$ in $E_c$, 
add the clause $(x_e \vee x_{e'})$ to $\phi$ if 
the vertices of $e$ and $e'$ induce a path of length 3 whose middle edge is in $F$ 
(see the forbidden configuration $(C)$ in Figure~\ref{fig:A-C}). 
\end{itemize}
Then, we obtain the bipartition $(E_r, E_b)$ of $E_c$ 
from a truth assignment $\tau$ of the variables as follows: 
\begin{itemize}
\item $x_e = 0$ in $\tau$ $\iff$ $e \in E_r$ 
(or $x_e = 1$ in $\tau$ $\iff$ $e \in E_b$). 
\end{itemize}
It is obvious that 
a truth assignment $\tau$ satisfies $\phi$ if and only if 
the corresponding bipartition of $E_c$ is $(A, C)$-free 
and all the committed edges in $F$ are in $E_b$. 

The 2CNF formula $\phi$ has at most $m$ Boolean variables. 
We can also see that 
$\phi$ has at most $f + 2 \cdot \min\{m^2, \hat{m}(\hat{m}+f)\}$ clauses 
since $\phi$ has at most two clauses for each pair of two edges in $E_c$ or 
for each pair of a non-edge in $\hat{E}$ and an edge in $F$. 
Then, $\phi$ can be obtained in $O(\min\{m^2, \hat{m}(\hat{m}+f)\})$ time 
from $G$ and $F$. 
Since a satisfying truth assignment of a 2CNF formula can be computed 
in linear time (see~\cite{APT79-IPL} for example), we have the following. 
\begin{lemma}
An $(A, C)$-free bipartition of $E_c$ 
can be computed in $O(\min\{m^2, \hat{m}(\hat{m}+f)\})$ time. 
\end{lemma}

\subsection{Swapping Edges}\label{sect:swap}
In this section, we show an $O(m\hat{m})$-time algorithm to transform 
a given $(A, C)$-free bipartition $(E_r, E_b)$ of $E_c$ 
into an $(A, B, C)$-free bipartition $(E_r', E_b')$ of $E_c$. 
For a non-edge $uv \in \hat{E}$, we define that 
\begin{align*}
H_r &= \{ u'v' \in E_r \mid uv', u'v \in E_b \}; \\
H_b &= \{ u'v' \in E_b \mid uv', u'v \in E_r \}; \\
H   &= H_r \cup H_b. 
\end{align*}
In other words, 
$H_r$ is the set of red edges of all configurations $(B_2)$ 
having non-edge $uv$, and 
$H_b$ is the set of blue edges of all configurations $(B_1)$ 
having non-edge $uv$. 
Between $E_r$ and $E_b$, we swap all edges in $H$ 
to obtain another bipartition $(E_r', E_b')$ of $E_c$, 
that is, we define that 
\begin{align*}
E_r' &= (E_r \setminus H_r) \cup H_b; \\
E_b' &= (E_b \setminus H_b) \cup H_r. 
\end{align*}
Since $(E_r, E_b)$ is $(A, C)$-free, we have $F \cap H = \emptyset$. 
Hence, all the committed edges in $F$ remain blue 
in the new bipartition $(E_r', E_b')$. 
Notice that by swapping the edges, 
we remove all the configurations $(B_1)$ and $(B_2)$ 
having non-edge $uv \in \hat{E}$. 
We claim that the swapping generates no forbidden configurations. 

\begin{figure*}[t]
  \psfrag{u}{$u$}
  \psfrag{v}{$v$}
  \psfrag{u1}{$u_1$}
  \psfrag{v1}{$v_1$}
  \psfrag{u2}{$u_2$}
  \psfrag{v2}{$v_2$}
  \psfrag{u3}{$u_2'$}
  \psfrag{v3}{$v_1'$}
  \psfrag{u4}{$u_2'$}
  \psfrag{v4}{$v_2'$}
  \centering
  \subcaptionbox{Case~1-1 in Lemma~\ref{lemma:swap} \label{fig:swap:1-1}}{\includegraphics[scale=.35]{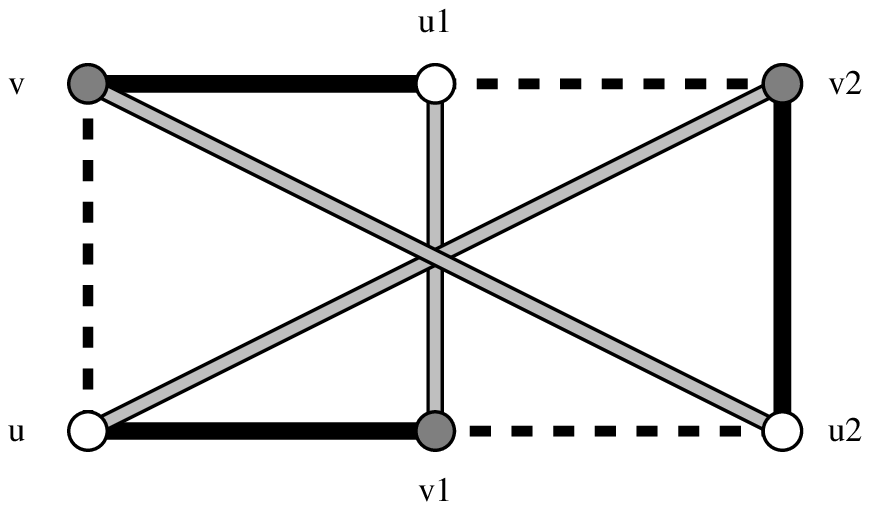}}
  \subcaptionbox{Case~3-1 in Lemma~\ref{lemma:swap} \label{fig:swap:3-1}}{\includegraphics[scale=.35]{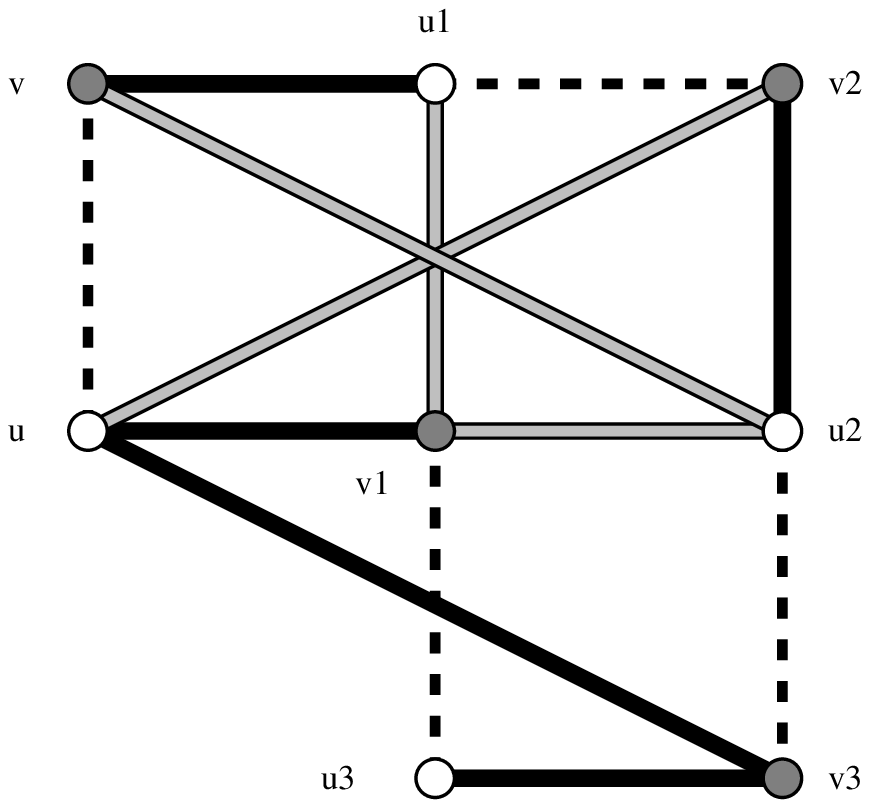}}
  \subcaptionbox{Case~3-3 in Lemma~\ref{lemma:swap} \label{fig:swap:3-3}}{\includegraphics[scale=.35]{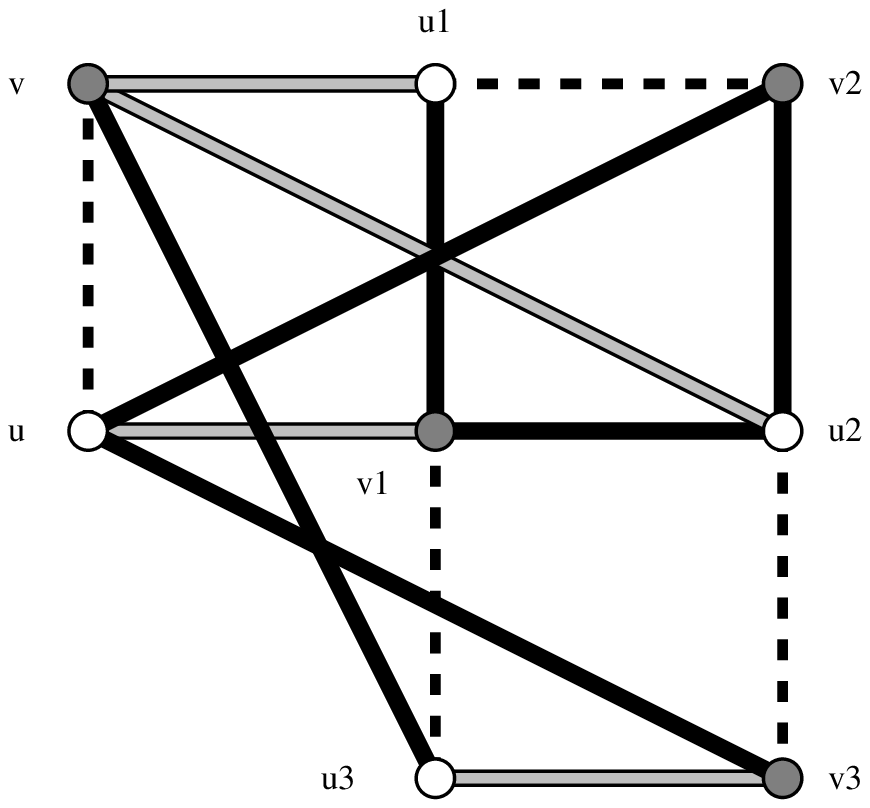}}
  \\
  \vspace{1em}
  \subcaptionbox{Case~5-1 in Lemma~\ref{lemma:swap} \label{fig:swap:5-1}}{\includegraphics[scale=.35]{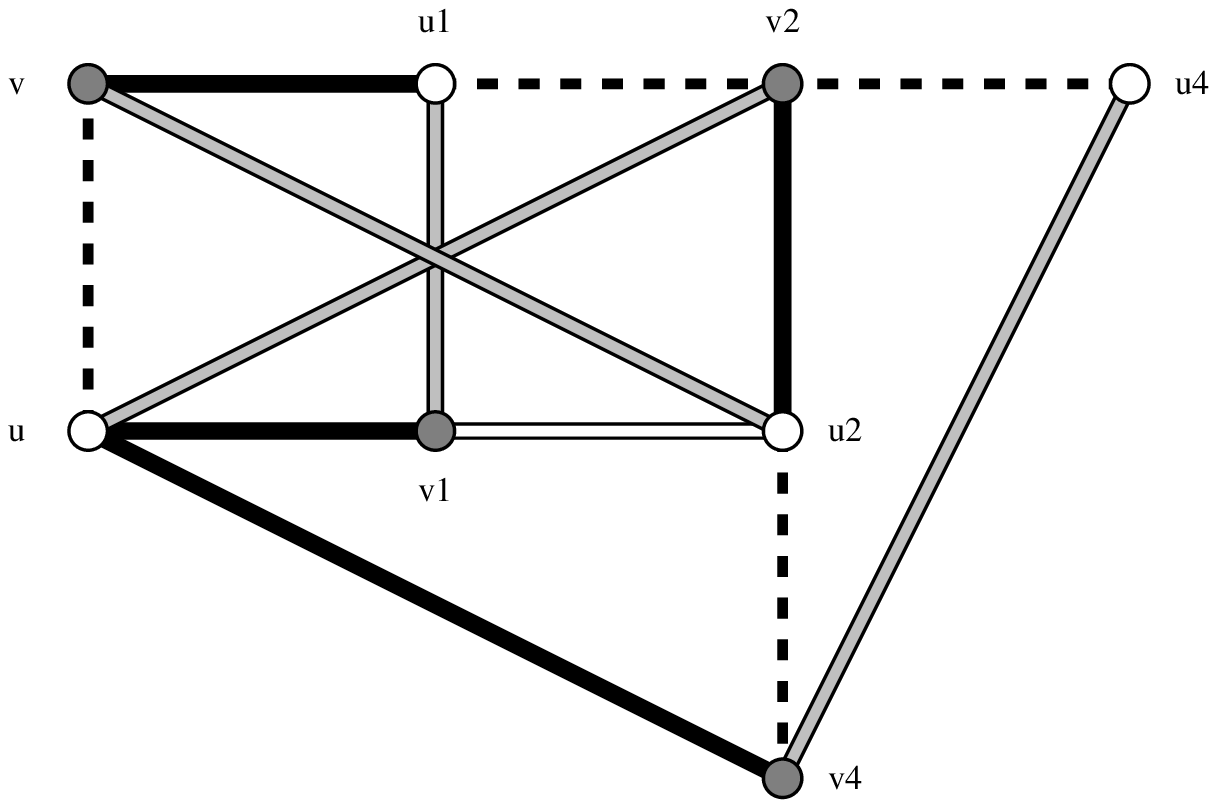}}
  \subcaptionbox{Case~2 in Lemma~\ref{lemma:G2} \label{fig:G2:2}}{\includegraphics[scale=.35]{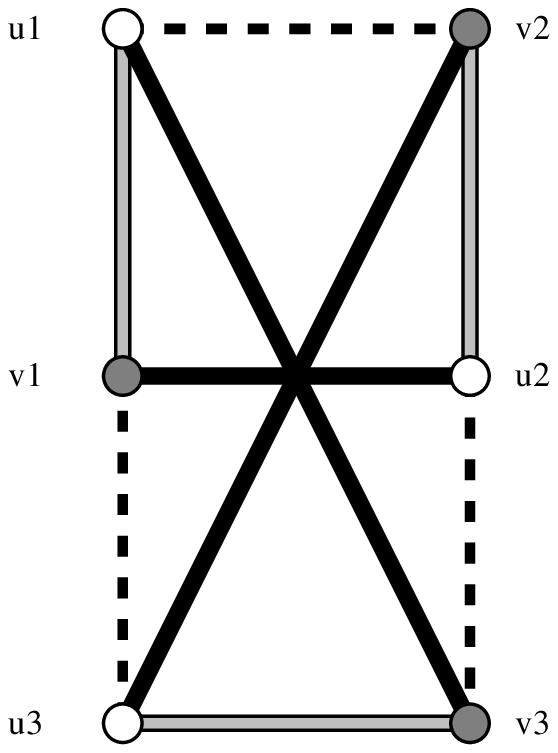}}
  \subcaptionbox{Case~3 in Lemma~\ref{lemma:G2} \label{fig:G2:3}}{\includegraphics[scale=.35]{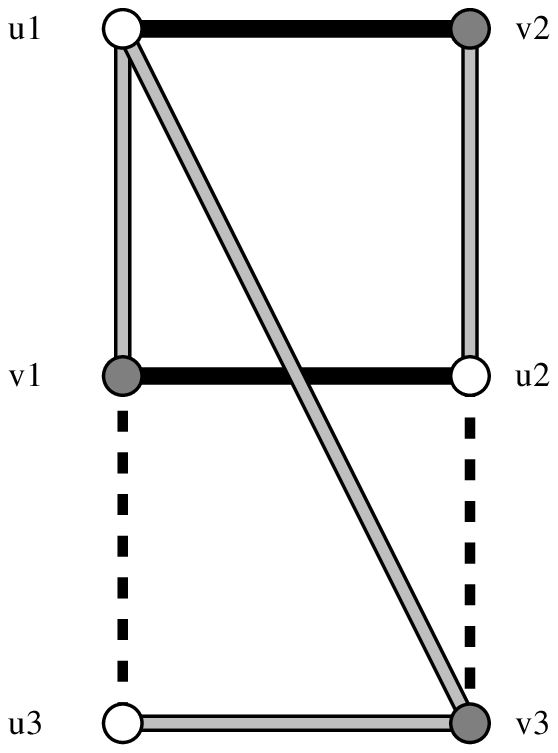}}
  \caption{
    Illustrating the proof of cases in Lemmas~\ref{lemma:swap} and~\ref{lemma:G2}. 
    Lines denote the same type of edges as in Figure~\ref{fig:A-C}. 
    } 
  \label{fig:proof}
\end{figure*}

\begin{lemma}\label{lemma:swap}
No edges in $H$ is an edge of any forbidden configurations 
of the new bipartition $(E_r', E_b')$ of $E_c$. 
\end{lemma}
\begin{proof}
We assume that the new bipartition $(E_r', E_b')$ has 
some configuration with at least one edge in $H$, 
and obtain a contradiction. 

\textbf{Case 1:} 
Suppose $(E_r', E_b')$ has a configuration $(A_1)$, 
that is, there are four vertices $u_1, v_1, u_2, v_2$ with 
$u_1v_1, u_2v_2 \in E_r'$ and $u_1v_2, u_2v_1 \in \hat{E}$. 

\textbf{Case 1-1:} 
Suppose $u_1v_1 \in H$ and $u_2v_2 \notin H$. 
This implies that $u_1v_1 \in E_b$ and $u_2v_2, uv_1, u_1v \in E_r$. 
See Figure~\ref{fig:proof}\subref{fig:swap:1-1}. 
We have $uv_2 \in E$, for otherwise 
$uv_1 \in E_r$ and $u_2v_2 \in E_r$ would be in conflict in $G$. 
Since $uv_2$ and $u_1v \in E_r$ are in conflict in $G$, we have $uv_2 \in E_b$. 
Similarly, we have $u_2v \in E$, for otherwise 
$u_1v \in E_r$ and $u_2v_2 \in E_r$ would be in conflict in $G$. 
Since $u_2v$ and $uv_1 \in E_r$ are in conflict in $G$, we have $u_2v \in E_b$. 
However, we have from $uv_2, u_2v \in E_b$ that 
$u_2v_2 \in H_r$, a contradiction. 

\textbf{Case 1-2:} 
Suppose $u_2v_2 \in H$ and $u_1v_1 \notin H$. 
This case is symmetric to Case~1-1. 

\textbf{Case 1-3:} 
Suppose $u_1v_1, u_2v_2 \in H$. 
This implies that $u_1v_1, u_2v_2 \in E_b$ and $uv_1, u_1v,$ $uv_2, u_2v \in E_r$, 
but $u_1v \in E_r$ and $uv_2 \in E_r$ are in conflict in $G$, a contradiction. 

\textbf{Case 2:} 
Suppose $(E_r', E_b')$ has a configuration $(A_2)$. 
This case is symmetric to Case~1. 

\textbf{Case 3:} 
Suppose $(E_r', E_b')$ has a configuration $(B_1)$, 
that is, there are four vertices $u_1, v_1, u_2, v_2$ with 
$u_1v_1, u_2v_2 \in E_r'$, $u_2v_1 \in E_b'$, and $u_1v_2 \in \hat{E}$. 

\textbf{Case 3-1:} 
Suppose $u_1v_1 \in H$ and $u_2v_2, u_2v_1 \notin H$. 
This implies that $u_1v_1, u_2v_1 \in E_b$ and $u_2v_2, uv_1, u_1v \in E_r$. 
See Figure~\ref{fig:proof}\subref{fig:swap:3-1}. 
We have $u_2v \in E$, for otherwise 
$u_1v \in E_r$ and $u_2v_2 \in E_r$ would be in conflict in $G$. 
If $u_2v \in E_r$, then we have from $uv_1 \in E_r$ that 
$u_2v_1 \in H_b$, a contradiction. 
Therefore, $u_2v \in E_b \cup E_u$. 
Since $u_2v_1 \in E_b$, there is an edge $u_2'v_1' \in E_r$ such that 
$u_2v_1$ and $u_2'v_1'$ are in conflict in $G$, 
that is, $u_2v_1', u_2'v_1 \in \hat{E}$. 
We have $uv_1' \in E$, for otherwise 
$uv_1 \in E_r$ and $u_2'v_1' \in E_r$ would be in conflict in $G$. 
Since $uv_1'$ and $u_2v \in E_b \cup E_u$ are in conflict in $G$, 
we have $uv_1' \in E_r$ and $u_2v \in E_b$. 
Then, we have $uv_2 \in E$, for otherwise 
$uv_1' \in E_r$ and $u_2v_2 \in E_r$ would be in conflict in $G$. 
Since $uv_2$ and $u_1v \in E_r$ are in conflict in $G$, we have $uv_2 \in E_b$. 
However, we have from $uv_2, u_2v \in E_b$ that 
$u_2v_2 \in H_r$, a contradiction. 

\textbf{Case 3-2:} 
Suppose $u_2v_2 \in H$ and $u_1v_1, u_2v_1 \notin H$. 
This case is symmetric to Case~3-1. 

\textbf{Case 3-3:} 
Suppose $u_2v_1 \in H$ and $u_1v_1, u_2v_2 \notin H$. 
This implies that $u_1v_1, u_2v_2, u_2v_1 \in E_r$ and $uv_1, u_2v \in E_b$. 
See Figure~\ref{fig:proof}\subref{fig:swap:3-3}. 
Since $u_2v_1 \in E_r$, there is an edge $u_2'v_1' \in E_b$ such that 
$u_2v_1$ and $u_2'v_1'$ are in conflict in $G$, 
that is, $u_2v_1', u_2'v_1 \in \hat{E}$. 
We have $uv_1' \in E$, for otherwise 
$uv_1 \in E_b$ and $u_2'v_1' \in E_b$ would be in conflict in $G$. 
Since $uv_1'$ and $u_2v \in E_b$ are in conflict in $G$, we have $uv_1' \in E_r$. 
Then, we have $uv_2 \in E$, for otherwise 
$uv_1' \in E_r$ and $u_2v_2 \in E_r$ would be in conflict in $G$. 
If $uv_2 \in E_b$, then we have from $u_2v \in E_b$ that 
$u_2v_2 \in H_r$, a contradiction. 
Therefore, $uv_2 \in E_r \cup E_u$. 
Similarly, we have $u_2'v \in E$, for otherwise 
$u_2v \in E_b$ and $u_2'v_1' \in E_b$ would be in conflict in $G$. 
Since $u_2'v$ and $uv_1 \in E_b$ are in conflict in $G$, we have $u_2'v \in E_r$. 
Then, we have $u_1v \in E$, for otherwise 
$u_1v_1 \in E_r$ and $u_2'v \in E_r$ would be in conflict in $G$. 
Since $u_1v$ and $uv_2 \in E_r \cup E_u$ are in conflict in $G$, 
we have $u_1v \in E_b$ and $uv_2 \in E_r$. 
However, we have from $uv_1 \in E_b$ that $u_1v_1 \in H_r$, a contradiction. 

\textbf{Case 3-4:} 
Suppose $u_1v_1, u_2v_2 \in H$ and $u_2v_1 \notin H$. 
We have a contradiction as Case 1-3. 

\textbf{Case 3-5:} 
Suppose $u_1v_1, u_2v_1 \in H$ and $u_2v_2 \notin H$. 
This implies that 
$u_1v_1 \in H_b$ and $u_2v_1 \in H_r$, but it follows that 
$uv_1 \in E_r$ from $u_1v_1 \in H_b$ and 
$uv_1 \in E_b$ from $u_2v_1 \in H_r$, a contradiction. 

\textbf{Case 3-6:} 
Suppose $u_2v_2, u_2v_1 \in H$ and $u_1v_1 \notin H$. 
This case is symmetric to Case~3-5. 

\textbf{Case 3-7:} 
Suppose $u_1v_1, u_2v_2, u_2v_1 \in H$. 
We have a contradiction as Case~3-5. 

\textbf{Case 4:} 
Suppose $(E_r', E_b')$ has a configuration $(B_2)$. 
This case is symmetric to Case~3. 

\textbf{Case 5:} 
Suppose $(E_r', E_b')$ has a configuration $(C)$, 
that is, there are four vertices $u_1, v_1, u_2, v_2$ with 
$u_1v_1, u_2v_2 \in E_r'$, $u_1v_2 \in \hat{E}$, and $u_2v_1 \in F$. 
Since the bipartition $(E_r, E_b)$ is $(A, C)$-free, we have $u_2v_1 \notin H$. 

\textbf{Case 5-1:} 
Suppose $u_1v_1 \in H$ and $u_2v_2 \notin H$. 
This implies that $u_1v_1 \in E_b$ and $u_2v_2, uv_1, u_1v \in E_r$. 
See Figure~\ref{fig:proof}\subref{fig:swap:5-1}. 
We have $uv_2 \in E$, for otherwise 
the vertices $u, v_1, u_2, v_2$ would induce a configuration $(C)$. 
Since $uv_2$ and $u_1v \in E_r$ are in conflict in $G$, we have $uv_2 \in E_b$. 
Similarly, we have $u_2v \in E$, for otherwise 
$u_1v \in E_r$ and $u_2v_2 \in E_r$ would be in conflict in $G$. 
If $u_2v \in E_r$, then 
the vertices $u, v_1, u_2, v$ would induce a configuration $(C)$. 
Therefore, $u_2v \in E_b \cup E_u$. 
Since $u_2v_2 \in E_r$, there is an edge $u_2'v_2' \in E_b$ 
such that $u_2v_2$ and $u_2'v_2'$ are in conflict in $G$, 
that is, $u_2v_2', u_2'v_2 \in \hat{E}$. 
We have $uv_2' \in E$, for otherwise $u_2'v_2' \in E_b$ and $uv_2 \in E_b$ 
would be in conflict in $G$. 
Since $uv_2'$ and $u_2v \in E_b \cup E_u$ are in conflict in $G$, 
we have $uv_2' \in E_r$ and $u_2v \in E_b$. 
However, we have from $uv_2, u_2v \in E_b$ that 
$u_2v_2 \in H_r$, a contradiction. 

\textbf{Case 5-2:} 
Suppose $u_2v_2 \in H$ and $u_1v_1 \notin H$. 
This case is symmetric to Case~5-1. 

\textbf{Case 5-3:} 
Suppose $u_1v_1, u_2v_2 \in H$. 
We have a contradiction as Case~1-3. 

Since all the cases above lead to contradictions, 
we conclude that the new bipartition $(E_r', E_b')$ has 
no forbidden configurations with an edge in $H$. 
\qed
\end{proof}

It follows from Lemma~\ref{lemma:swap} that 
continuing in this way for each non-edge in $\hat{E}$, 
we can obtain an $(A, B, C)$-free bipartition of $E_c$. 
Since the set $H$ can be computed in $O(m)$ time 
for each non-edge in $\hat{E}$, 
the overall running time is $O(m\hat{m})$. 
\begin{lemma}
From a given $(A, C)$-free bipartition of $E_c$, 
an $(A, B, C)$-free bipartition of $E_c$ 
can be computed in $O(m\hat{m})$ time. 
\end{lemma}

\subsection{Adding edges}\label{sect:add}
In this section, we claim that 
a given $(A, B, C)$-free bipartition $(E_r, E_b)$ of $E_c$ can be extended 
in linear time into a 2-chain subgraph cover $(G_1, G_2)$ of $G$ 
such that $G_1$ has no edges in $F$. 
We first show the following. 
\begin{lemma}\label{lemma:G2}
The subgraph of $G$ induced by $E_b \cup E_u$ is a chain graph. 
\end{lemma}
\begin{proof}
We show that no $2K_2$ is in the subgraphs of $G$ induced by $E_b \cup E_u$. 

\textbf{Case~1:} Suppose $u_1v_1, u_2v_2 \in E_b \cup E_u$ and 
$u_1v_2, u_2v_1 \in \hat{E}$. 
It is obvious that $u_1v_1, u_2v_2 \notin E_u$, 
but $u_1v_1, u_2v_2 \in E_b$ implies that 
the vertices $u_1, v_1, u_2, v_2$ induce a configuration $(A_2)$, 
a contradiction. 

\textbf{Case~2:} Suppose $u_1v_1, u_2v_2 \in E_b \cup E_u$, 
$u_1v_2 \in \hat{E}$, and $u_2v_1 \in  E \setminus (E_b \cup E_u)$. 
Since $u_2v_1 \in E \setminus (E_b \cup E_u) = E_r$, there is an edge 
$u_2'v_1' \in E_b$ such that $u_2v_1$ and $u_2'v_1'$ are in conflict in $G$, 
that is, $u_2'v_1, u_2v_1' \in \hat{E}$. 
See Figure~\ref{fig:proof}\subref{fig:G2:2}. 
We have $u_1v_1' \in E$, for otherwise the vertices 
$u_1, v_1, u_2', v_1'$ would induce a configuration in Case~1. 
Since $u_1v_1'$ and $u_2v_2 \in E_b \cup E_u$ are in conflict in $G$, 
we have $u_2v_2 \in E_b$. 
Similarly, we have $u_2'v_2 \in E$, for otherwise the vertices 
$u_2, v_2, u_2', v_1'$ would induce a configuration in Case~1. 
Since $u_2'v_2$ and $u_1v_1 \in E_b \cup E_u$ are in conflict in $G$, 
we have $u_1v_1 \in E_b$, 
but then the vertices $u_1, v_1, u_2, v_2$ induce a configuration $(B_2)$, 
a contradiction. 

\textbf{Case~3:} Suppose $u_1v_1, u_2v_2 \in E_b \cup E_u$ 
and $u_1v_2, u_2v_1 \in E \setminus (E_b \cup E_u)$. 
Since $u_2v_1 \in E \setminus (E_b \cup E_u) = E_r$, there is an edge 
$u_2'v_1' \in E_b$ such that $u_2v_1$ and $u_2'v_1'$ are in conflict in $G$, 
that is, $u_2'v_1, u_2v_1' \in \hat{E}$. 
See Figure~\ref{fig:proof}\subref{fig:G2:3}. 
We have $u_1v_1' \notin \hat{E} \cup E_r$, for otherwise the vertices 
$u_1, v_1, u_2', v_1'$ would induce a configuration in Case~1 or Case~2. 
However, $u_1v_1' \in E_b \cup E_u$ implies that 
the vertices $u_2, v_2, u_1, v_1'$ induce 
a configuration in Case~2, a contradiction. 

Since all the cases above lead to contradictions, 
we conclude that the subgraph of $G$ induced by $E_b \cup E_u$ has 
no $2K_2$, and it is a chain subgraph of $G$. 
\qed
\end{proof}

We next show that $E_r$ can be extended into a chain graph in $G - F$, 
the subgraph of $G$ obtained by removing all the edges in $F$. 
To do this, we consider the following problem: 
Given a graph $H$ and a set $M$ of edges of $H$, 
find a chain subgraph $C$ of $H$ containing all edges in $M$. 
This problem is called the \emph{chain graph sandwich problem}, and 
the chain graph $C$ is called a \emph{chain completion} of $M$ in $H$. 
Although the chain graph sandwich problem is NP-complete, 
it can be solved in linear time 
if $H$ is a bipartite graph~\cite{DFGKM11-AOR}. 
The chain graph sandwich problem on bipartite graphs is closely related to 
the threshold graph sandwich problem~\cite{GKS95-JAL,RS95-STOC} 
(see also Section 1.5 of~\cite{MP95-book}), 
and in the proof of Lemma~\ref{lemma:sandwich}, 
we will use an argument similar to that used in the literature. 

Let $H = (U, V, E)$ be a bipartite graph, 
let $\hat{E}$ be the set of edges of the bipartite complement $\hat{H}$ of $H$, 
and let $k \geq 2$. 
A set of $k$ distinct vertices $u_0, u_1, \ldots, u_{k-1}$ in $U$ and 
$k$ distinct vertices $v_0, v_1, \ldots, v_{k-1}$ in $V$ is called 
an \emph{alternating cycle of $M$ relative to $H$} 
if $u_iv_i \in \hat{E}$ and $u_{i+1}v_i \in M$ for any $i$, $0 \leq i < k$ 
(indices are modulo $k$). 
Note that an alternating cycle of $M$ with lengh 4 relative to $H$ 
is exactly a $2K_2$ of $M$ in $H$. 
\begin{lemma}~\label{lemma:sandwich}
Let $M$ be a set of edges in a bipartite graph $H$. 
\begin{itemize}
\item The set $M$ of edges has a chain completion in $H$ 
if and only if there are no alternating cycles of $M$ relative to $H$. 
\item The chain completion of $M$ in $H$ can be computed in $O(n+m)$ time. 
\end{itemize}
\end{lemma}
\begin{proof}
The proof is in Appendix. 
The details of the algorithm are also shown in~\cite{DFGKM11-AOR}. 
\qed
\end{proof}

Then, we show that $E_r$ has a chain completion in $G - F$. 
\begin{lemma}~\label{lemma:G1}
There are no alternating cycles of $E_r$ relative to $G - F$. 
\end{lemma}
\begin{proof}
We first prove that 
there are no alternating cycles of $E_r$ with length 4 relative to $G - F$, 
that is, no two edges in $E_r$ are in conflict in $G - F$. 
Since the bipartition $(E_r, E_b)$ 
does not have a configuration $(A_1)$ or $(C)$, 
it is enough to show that 
$(E_r, E_b)$ has no configuration consisting of 
four vertices $u_1, v_1, u_2, v_2$ 
with edges $u_1v_1, u_2v_2 \in E_r$ and $u_1v_2, u_2v_1 \in F$. 
Suppose $(E_r, E_b)$ has such a configuration. 
Since $u_1v_1 \in E_r$, there is an edge 
$u_1'v_1' \in E_b$ such that $u_1v_1$ and $u_1'v_1'$ are in conflict in $G$, 
that is, $u_1v_1', u_1'v_1 \in \hat{E}$. 
We have $u_2v_1' \in E$, for otherwise 
$u_2v_1 \in F$ and $u_1'v_1' \in E_b$ would be in conflict in $G$ 
(recall that $F \subseteq E_b \cup E_u$). 
If $u_2v_1' \in E_r$, then the vertices $u_1, v_1, u_2, v_1'$ 
would induce a configuration $(C)$, a contradiction. 
Therefore, $u_2v_1' \in E_b \cup E_u$. 
Similarly, since $u_2v_2 \in E_r$, there is an edge 
$u_2'v_2' \in E_b$ such that $u_2v_2$ and $u_2'v_2'$ are in conflict in $G$, 
that is, $u_2v_2', u_2'v_2 \in \hat{E}$. 
We have $u_1v_2' \in E$, for otherwise 
$u_1v_2 \in F$ and $u_2'v_2' \in E_b$ would be in conflict in $G$. 
Since $u_1v_2'$ and $u_2v_1' \in E_b \cup E_u$ are in conflict in $G$, 
we have $u_1v_2' \in E_r$ and $u_2v_1' \in E_b$. 
However, the vertices $u_2, v_2, u_1, v_2'$ induce a configuration $(C)$, 
a contradiction. 
Thus, there are no alternating cycles of $E_r$ 
with length 4 relative to $G - F$. 

We now suppose that there are an alternating cycle of $E_r$ 
with length grater than 4 relative to $G - F$. 
Let $AC$ be such an alternating cycle with minimal length, and 
let $u_0, v_0, u_1, v_1, \ldots u_{k-1}, v_{k-1}$ be 
the consecutive vertices of $AC$ 
with $u_iv_i \in \hat{E} \cup F$ and $u_{i+1}v_i \in E_r$ 
for any $i$, $0 \leq i < k$ (indices are modulo $k$). 

We claim that $AC$ has no edges in $F$. 
Suppose $u_1v_1 \in F$. 
We have $u_2v_0 \in E$, for otherwise the vertices 
$u_2, v_1, u_1, v_0$ would induce a configuration $(C)$. 
If $u_2v_0 \in E_r$, then the vertices 
$u_0, v_0, u_2, v_2, \ldots u_{k-1}, v_{k-1}$ form 
a shorter alternating cycle of $E_r$ relative to $G - F$, 
contradicting the minimality of $AC$. 
Therefore, $u_2v_0 \in E_b \cup E_u$. 
Since $u_1v_0 \in E_r$, there is an edge 
$u_1'v_0' \in E_b$ such that 
$u_1v_0$ and $u_1'v_0'$ are in conflict in $G$, that is, 
$u_1'v_0, u_1v_0' \in \hat{E}$. 
Similarly, since $u_2v_1 \in E_r$, there is an edge 
$u_2'v_1' \in E_b$ such that 
$u_2v_1$ and $u_2'v_1'$ are in conflict in $G$, that is, 
$u_2'v_1, u_2v_1' \in \hat{E}$. 
The edges $u_1'v_0'$ and $u_2'v_1'$ are not the same edge, 
for otherwise $u_1'v_0' \in E_b$ and $u_2v_0 \in E_b \cup E_u$ 
would be in conflict in $G$. 
Then, the vertices $u_1, v_1, u_2', v_1', u_2, v_0, u_1', v_0'$ form 
an alternating cycle of $E_b \cup E_u$ relative to $G$ 
(recall that $F \subseteq E_b \cup E_u$). 
It follows from Lemma~\ref{lemma:sandwich} that 
$E_b \cup E_u$ does not induce a chain graph, 
contradicting Lemma~\ref{lemma:G2}. Thus, $AC$ has no edges in $F$. 

Recall that the length of $AC$ is at least 6, and 
let $u_0, v_0, u_1, v_1, u_2, v_2$ denote the consecutive vertices of $AC$. 
Since $AC$ has no edges in $F$, we have 
$u_0v_0, u_1v_1, u_2v_2 \in \hat{E}$ and $u_1v_0, u_2v_1\in E_r$. 
We have $u_2v_0 \in E$, for otherwise 
$u_1v_0 \in E_r$ and $u_2v_1 \in E_r$ would be in conflict in $G$. 
If $u_2v_0 \in E_r$, then the vertices 
$u_0, v_0, u_2, v_2, \ldots u_{k-1}, v_{k-1}$ form 
a shorter alternating cycle of $E_r$ relative to $G - F$, 
contradicting the minimality of $AC$. 
Therefore, $u_2v_0 \in E_b \cup E_u$. 
On the other hand, 
if $u_0v_1 \in \hat{E}$, then the vertices 
$u_0, v_1, u_2, v_2, \ldots u_{k-1}, v_{k-1}$ form 
a shorter alternating cycle of $E_r$ relative to $G - F$, 
contradicting the minimality of $AC$. Therefore, $u_0v_1 \in E$. 
Since $u_0v_1$ and $u_1v_0 \in E_r$ are in conflict in $G$, we have $u_0v_1 \in E_b$. 
By similar arguments, we have $u_1v_2 \in E_b$. 
Then, we have $u_0v_2 \in E$, for otherwise 
$u_0v_1 \in E_b$ and $u_1v_2 \in E_b$ would be in conflict in $G$. 
Since $u_0v_2$ and $u_2v_0 \in E_b \cup E_u$ are in conflict in $G$, 
we have $u_0v_2 \in E_r$ and $u_2v_0 \in E_b$. 
This implies that the vertices 
$u_1, v_0, u_2, v_1$ induce a configurations $(B_1)$, a contradiction. 

Thus, we conclude that 
there are no alternating cycles of $E_r$ relative to $G - F$. 
\qed
\end{proof}

Now, we have the following from Lemmas~\ref{lemma:sandwich} and~\ref{lemma:G1}. 
\begin{lemma}\label{lemma:G1-2}
There is a chain completion of $E_r$ in $G - F$, and 
it can be computed in linear time from $E_r$. 
\end{lemma}

Since every edge of $G$ belongs to either $E_r$ or $E_b \cup E_u$, 
$G$ can be covered by the chain completion of $E_r$ in $G - F$ and 
the chain subgraph of $G$ induced by $E_b \cup E_u$. 
Thus, we have the following from Lemmas~\ref{lemma:G2} and~\ref{lemma:G1-2}. 
\begin{lemma}
From a given $(A, B, C)$-free bipartition of $E_c$, 
a 2-chain subgraph cover $(G_1, G_2)$ of $G$ 
such that $G_1$ has no edges in $F$ 
can be computed in linear time. 
\end{lemma}

\section{Concluding Remarks}
This paper provides an $O(m\hat{m} + \min\{m^2, \hat{m}(\hat{m}+f)\})$-time algorithm 
to solve the restricted 2-chain subgraph cover problem by reducing it to 2SAT. 
To do this, we show that the problem has a feasible solution if and only if 
there is an $(A, C)$-free bipartition of 
the set of committed edges of the given bipartite graph. 
This result implies a simpler recognition algorithm 
for simple-triangle graphs. 

We finally note that for simple-triangle graphs, 
structure characterizations as well as 
the complexity of the graph isomorphism problem 
still remain open questions. 

\subsection*{Acknowledgments}
We are grateful to anonymous referees 
for careful reading and helpful comments. 
A part of 
this work was done while the author was in Tokyo Institute of Technology and 
supported by JSPS Grant-in-Aid for JSPS Fellows (26$\cdot$8924). 
The final publication is available at Springer via~\url{http://dx.doi.org/10.1007/978-3-319-53925-6_14}. 


\appendix

\section{Proof of Lemma~\ref{lemma:sandwich}}
\begin{proof}
We first prove the "only-if" part. 
A vertex of a bipartite graph is called an \emph{isolated vertex} 
if it is not adjacent to any vertex, 
and a vertex is called a \emph{dominating vertex} 
if it is adjacent to all the vertices 
on the other side of the bipartition. 
It is known that a chain graph has an isolated vertex or 
a dominating vertex~\cite{MP95-book}. 
Suppose $M$ has a chain completion $C$ in $H$, and 
there is an alternating cycle of $M$ relative to $H$. 
Since $C$ has all the edges in $M$, 
it also has all the vertices on the alternating cycle. 
Let $C'$ be the subgraph of $C$ induced by 
the vertices on the alternating cycle. 
Since $C$ is a chain graph, $C'$ is also a chain graph. 
However, $C'$ has neither isolated vertex nor dominating vertex 
since any vertex of $C'$ is incident to an edge in $M$ 
and incident to a non-edge in $\hat{E}$. 
It follows that $C'$ is not a chain graph, a contradiction. 

We next prove the "if" part by induction. We assume that 
the lemma holds for any bipartite graph with fewer vertices than $H$. 
Suppose there are no alternating cycles of $M$ relative to $H$. 
Then, $H$ has a vertex incident to no edges in $M$ or 
a vertex incident to no edges in $\hat{E}$, 
for otherwise we can grow a path alternating between $M$ and $\hat{E}$ 
until an alternating cycle is obtained. 
Let $u$ be such a vertex, and 
we assume without loss of generality that $u \in U$. 
Let $M - u$ be the set of edges in $M$ not incident to $u$, and 
let $H - u$ be the subgraph of $H$ obtained by removing $u$. 
Since there are no alternating cycles of $M$ relative to $H$, 
there are no alternating cycles of $M - u$ relative to $H - u$. 
It follows by induction that 
there is a chain completion $C = (U, V, E)$ of $M - u$ in $H - u$. 
If $u$ is incident to no edges in $M$, 
then $C' = (U \cup \{u\}, V, E)$ 
is a chain completion of $M$ in $H$. 
If $u$ is incident to no edges in $\hat{E}$, 
then $C' = (U \cup \{u\}, V, E \cup \{uv \mid v \in V \})$ 
is a chain completion of $M$ in $H$. 

The proof of "if" part gives a linear-time algorithm that 
finds a chain completion of $M$ in $G$. 
The details of the algorithm are also shown in~\cite{DFGKM11-AOR}. 
\qed
\end{proof}

\end{document}